\documentclass[letterpaper,11pt]{article}
\usepackage[T1]{fontenc}
\usepackage[utf8]{inputenc}
\usepackage{lmodern}
\usepackage{microtype}
\usepackage{complexity}
\usepackage{amsmath}
\usepackage{amssymb}
\usepackage{amsthm}
\usepackage[top=1in,bottom=1in,left=1in,right=1in]{geometry}
\usepackage[textsize=tiny]{todonotes}
\usepackage{graphicx}
\usepackage[export]{adjustbox}
\usepackage{wrapfig}
\usepackage{xspace}
\usepackage[unicode=true]{hyperref}
\usepackage[linesnumbered,ruled]{algorithm2e}
\usepackage[noend]{algorithmic}
\newtheorem{thm}{Theorem}
\newtheorem{lem}{Lemma}

\theoremstyle{definition}

\newcommand{\OTSP}{\textnormal{OTSP}\xspace}
\newcommand{\PTSP}{\textnormal{PTSP}\xspace}
\newcommand{\TSP}{\textnormal{TSP}\xspace}
\newcommand{\kTSPP}{\textnormal{$k$-TSPP}\xspace}

\newcommand{\drop}{\textnormal{drop}\xspace}

\title{Approximating Traveling Salesman Problems Using a Bridge Lemma}
\author{
    \begin{minipage}[t]{.3\textwidth}
        \centering
        Martin Böhm\thanks{Research supported by National Science Centre in Poland under grant SONATA 2022/47/D/ST6/02864.}\\
        University of Wroc{\l}aw
    \end{minipage}
    \begin{minipage}[t]{.3\textwidth}
        \centering
        Zachary Friggstad\thanks{Research supported by an NSERC Discovery Grant.}\\
        University of Alberta
    \end{minipage}\\[15mm]
    \begin{minipage}[t]{.3\textwidth}
        \centering
        Tobias Mömke\thanks{Partially supported by DFG Grant 439522729 (Heisenberg-Grant) and DFG Grant 439637648 (Sachbeihilfe).}\\
        University of Augsburg
    \end{minipage}
    \begin{minipage}[t]{.3\textwidth}
        \centering
        Joachim Spoerhase\thanks{Research partially supported by European Research Council (ERC) under the European Union’s Horizon 2020 research and innovation programme (grant agreement No. 759557).}\\
        University of Sheffield
    \end{minipage}\\[15mm]
}
\begin{document}
\maketitle
\thispagestyle{empty}
\begin{abstract}
We give improved approximations for two metric \textsc{Traveling Salesman Problem} (\TSP) variants. In \textsc{Ordered TSP} (\OTSP) we are given a linear ordering on a subset of nodes $o_1, \ldots, o_k$. The TSP solution must have that $o_{i+1}$ is visited at some point after $o_i$
for each $1 \leq i < k$. This is the special case of \textsc{Precedence-Constrained TSP} ($\PTSP$) in which the precedence constraints are given by a single chain on a subset of nodes. In \textsc{$k$-Person TSP Path} (\kTSPP), we are given pairs of nodes $(s_1,t_1), \ldots, (s_k,t_k)$. The goal is to find an $s_i$-$t_i$ path with minimum total cost such that every node is visited by at least one path.

We obtain a $3/2 + e^{-1} < 1.878$ approximation for \OTSP, the first improvement over a trivial $\alpha+1$ approximation where $\alpha$ is the current best TSP approximation. We also obtain a $1 + 2 \cdot e^{-1/2} < 2.214$ approximation for \kTSPP, the first improvement over a trivial $3$-approximation.

These algorithms both use an adaptation of the Bridge Lemma that was initially used to obtain improved \textsc{Steiner Tree} approximations [Byrka et al., 2013]. Roughly speaking, our variant states that the cost of a cheapest forest rooted at a given set of terminal nodes will decrease by a substantial amount if we randomly sample a set of non-terminal nodes to also become terminals such provided each non-terminal has a constant probability of being sampled. We believe this view of the Bridge Lemma will find further use for improved vehicle routing approximations beyond this paper.
\end{abstract}
\newpage
\setcounter{page}{1}

\section{Introduction}

In this work, we investigate two established variants of the quintessential \emph{metric traveling salesman problem}. As the name suggests, we are working in a metric space $(V,c)$ with a universe $V$ of $n$ vertices and a metric $c\colon V \times V \rightarrow \mathbb Q_{\geq 0}$, which we will also call a \emph{cost function}. We make use of the fact that a metric induces a complete graph $K_n$ with a cost $c(u,v)$ on each edge $\{u,v\}$, so that we can talk of paths and cycles in our metric space.

In the first problem, called \textsc{Ordered Traveling Salesman Problem} (OTSP), we are given a metric $(V,c)$ alongside a sequence of distinct nodes $o_1, o_2, \ldots, o_k \in V$. The goal is to find a minimum-cost Hamiltonian cycle such that a traversal of this cycle visits $o_{i+1}$ after $o_i$ for each $1 \leq i < k$.

We emphasize that vertices not present in $\{o_1, \ldots, o_k\}$ must be visited also, but their order is unrestricted, and that many other vertices can be visited between $o_i$ and $o_{i+1}$. 

In the second problem, called \textsc{$k$-Person TSP Path} (\kTSPP), we are presented with the metric $(V,c)$ and a list of $k$ vertex pairs $(s_1, t_1), \ldots, (s_k, t_k)$. The goal is to find one $s_i$-$t_i$ path for each $1 \leq i \leq k$ having minimum total cost such that each vertex of $V$ lies on at least one path.

By creating colocated copies of locations, we may assume these $2k$  endpoints are distinct and, by shortcutting, that each vertex lies on precisely one path.

\subsection{Previous Work}
\OTSP is a special case of \textsc{Priority TSP} (\PTSP) in which we are given a collection of precedence constraints $(u_1, v_1), (u_2, v_2), \ldots, (u_k, v_k)$. The goal is to find a minimum-cost Hamiltonian cycle such that some direction of traversal has $u_i$ being visited before $v_i$ for each $1 \leq i \leq k$. Charikar et al.~\cite{CMRS97} show that approximating $\PTSP$ with general partial orderings determined by the vertex pairs is hard.
More specifically, they show that it is $\NP$-hard to obtain a constant approximation and that a polylogarithmic approximation for $\PTSP$ would imply $\NP \subseteq \DTIME(n^{\log\log n})$. These results even hold if the graph is a path metric (i.e., the metric completion of a simple path).

The problem $\OTSP$ and thus $\PTSP$ contain standard (metric) \textsc{TSP} as
a subproblem, for which the best approximation ratio $\alpha$ is currently a constant slightly smaller than $1.5$ \cite{OG1,OG2}. This landmark result was the first improvement over the classic $1.5$-approximation
Christofides and Serdyukov approximation from the 1970s \cite{Chr76,Ser78}.
Prior to this paper, the best constant approximation ratio for $\OTSP$ is obtained by using an $\alpha$-approximation for \TSP as a black box.
Namely, compute a $\alpha$-approximate TSP solution skipping the vertices $\{o_1,\dotsc,o_k\}$ and combine the obtained cycle with the cycle $o_1,o_2,\dotsc,o_k,o_1$ to obtain an $(\alpha+1)$-approximate solution~\cite{BHKK06}.
Additionally, another sub-constant improvement was known that is slightly better than $\alpha + 1$ for modest values of $k$: there is a $(2.5 - 2/k)$-approximation algorithm for $\OTSP$~\cite{BMS13}.

$\OTSP$ was studied in the context of cost given by a generalized metric, where the triangle inequality can be slightly violated \cite{BHKK06,BMS13,BS14}.
There is a close relation between $\OTSP$ and vehicle routing problems with pickup and delivery~\cite{SS95,BCGL07} and with the dial-a-ride problem~\cite{CR98,GNR09,CL07}.

\PTSP has a long history of research and practical applications, cf.~\cite{BMRS94,Sal19} and references within.
There are, however, only few theoretical results.
The disparity can be partially explained by the strong hardness results obtained by Bhatida, Khuller and Naor~\cite{BKN00} for the closely related
load time scheduling problem (LTSP) and the subsequent hardness results for $\PTSP$ of Charikar et al.~\cite{CMRS97}, already mentioned in the beginning of the introduction.\footnote{The authors presented a preliminary version of the result in a short talk at the Aussois Combinatorial Optimization Workshop in January 2024, aiming to resolve a final step of the proof. The goal was achieved during the workshop by adding Theorem~\ref{thm:bj}, obtaining a better-than-2.5 approximation for $\OTSP$.
After obtaining all remaining results, the authors learned about the result~\cite{AMN24} on \OTSP.}

Our results are based on applying an iterative linear programming technique.
It is therefore reasonable to check, whether other iterated methods could work, based on the framework of Jain~\cite{Jai01}
and the subsequent development~\cite{LRS11}.
Applying these techniques requires a linear programming formulation based on skew-supermodular cut functions.
While it is not hard to find relaxations for $\OTSP$ based on cut functions, obtaining a formulation where these functions are skew-supermodular is an open problem.
We can therefore also not directly apply well-known primal-dual frameworks~\cite{GW95,WGMV95}.

Regarding \kTSPP, a simple 3-approximation is known by finding the cheapest forest where each component contains one of the $2k$ endpoints (by contracting all $2k$ endpoints and computing a minimum spanning tree), doubling the edges of this forest, adding the edges $\{(s_i,t_i) : 1 \leq i \leq k\}$ and shortcutting the resulting Eulerian $s_i$-$t_i$ walks. To the best of our knowledge, this algorithm is folklore.
The special case $k = 1$, i.e., \textsc{TSP Path} (TSPP), is very well-studied. Recently, it was shown to that an $\alpha$-approximation for \TSP yields an $(\alpha + \epsilon)$-approximation for (TSPP) \cite{PTSP22}.

When $k > 1$, the closest work is on the case $s_i = t_i$, i.e. we are given $k$ root nodes and each cycle must be rooted at one of the nodes. The 3-approximation for \kTSPP  described above can be adapted easily to get a 2-approximation for this special case which is also the best known in general. Xu and Rodrigues obtain a $1.5$-approximation, matching the approximation guarantee of the Christofides/Serdyukov algorithm, for this special case running in time $n^{O(k)}$ \cite{Xu10}. More recently, a $(1.5+\epsilon)$-approximation for this variant was given by Deppert, Kaul, and Mnich \cite{Depp23} with an FPT-like running time on $k$, i.e. $O((1/\epsilon)^{O(k \log k)} \cdot {\rm poly}(n))$. It is conceivable that ideas in \cite{Xu10,Depp23} could lead to a better-than-3 approximation for \kTSPP in the case where $k$ is a constant, but this is not our focus.

Finally, Giannako et al.~\cite{Gian17} study the variant of \kTSPP where only the start locations are given (the $k$ paths may end anywhere) and obtain various non-trivial constant factor approximations in the case where $G$ is the shortest path metric of an undirected, cubic, and 2-vertex connected graph as well as other similar results for certain graph classes. We remark this version (given start points) has a trivial $2$-approximation by simply computing the minimum-cost spanning forest having each component rooted in a start node and then doubling the corresponding trees to one path (or one cycle) starting at each start node.


\subsection{Our Results}
We obtain improved approximations for both \OTSP and \kTSPP using a randomized LP rounding approach for both. For \kTSPP we consider writing an LP relaxation that models one unit of $s_i$-$t_i$ flow while ensuring each node in $G$ supports at least one unit of flow in total. Using a decomposition result of Bang-Jensen, we sample an $s_i$-$t_i$ path for each $1 \leq i \leq k$ such that every node in $G$ has a constant probability of being covered. Our key insight is then to use an adaptation of the bridge lemma of Byrka et al.~\cite{BGRS13} used in improved Steiner Tree approximations to show that the cost of attaching the remaining nodes using a doubled minimum-cost forest rooted at the already-covered nodes is, in expectation, a constant factor smaller than the optimum solution cost: small enough that even doubling the forest to graft in the remaining nodes to get the final paths still yields an improved approximation.
\begin{thm}\label{thm:ktspp}
There is a polynomial-time randomized algorithm for \kTSPP that finds a solution whose expected cost is at most $1 + 2 \cdot e^{-1/2} < 2.214$ times the optimum solution cost.
\end{thm}

Our \OTSP improvement is similar in that we view the problem as a special case of \kTSPP where the pairs are given as $\{(o_i,o_{i+1}) : 1 \leq i < k\} \cup \{(o_k, o_1)\}$. Specific properties of this special case yield even better approximations.
\begin{thm}\label{thm:otsp}
There is a polynomial-time randomized algorithm for \OTSP that finds a solution whose expected cost is at most $(3/2 + e^{-1}) < 1.868$ times the optimum solution cost.
\end{thm}

As mentioned earlier, a key part of our improvements comes from an adaptation of the Bridge Lemma from \cite{BGRS13}. In particular, we show that the following holds.
For a (not necessarily metric) graph $G = (V,E)$ with edge costs $c_e \geq 0$, and a subset $T \subseteq V$ of nodes, let $c_T$ denote the cheapest spanning forest such that each component contains exactly one node in $T$: call this a $T$-rooted spanning forest.

\begin{thm}\label{thm:bridge}
Let $\mu : 2^V \rightarrow \mathbb{Q}_{\geq 0}$ be a probability distribution over subsets of $V \setminus T$ and let $\gamma \geq 0$ be such that for any $v \in V\setminus T, {\bf Pr}_{S \sim \mu}[v \notin S] \leq \gamma$. Then ${\bf E}[c_{T \cup S}] \leq \gamma \cdot c_T$.
\end{thm}
Intuitively, if the sampling procedure has each $v \in V\setminus T$ being included with constant probability, then sampling a single set $S$ and adding it to $T$ causes the minimum $T$-rooted spanning forest cost to drop by a constant factor in expectation.


\subsection{Overview of our Algorithm}
For brevity, we let $T = \{s_1, t_1, s_2, t_2, \ldots, s_k, t_k\}$ be the set of all endpoints of paths. Recall that we may assume all endpoints are distinct, so $|T| = 2k$.

We consider a linear programming relaxation for \kTSPP, namely \eqref{lp:ktspp} defined below. First, let $G'$ denote the bidirected version of $G$. That is, $G'$ is a complete directed graph where for each undirected edge $e = \{u,v\}$ we have both arcs $(u,v)$ and $(v,u)$ appearing in $G'$ whose costs are the same as $e$. For an arc $a = (u,v)$, we let $c(a)$ denote the cost of the underlying edge $\{u,v\}$. 

The LP relaxation has a variable $x_{i,a}$ for each $i \in [k]$ and each arc $a$ indicating if the $s_i$-$t_i$ path traverses the corresponding edge in the direction of $a$ plus a variable $z_{i,v}$ for each $i \in [k]$ and each $v \in V\setminus T$ indicating $v$ will be covered by path $i$. See the next section for descriptions of some notation.

\begin{alignat}{2}
    {\bf minimize}:\quad \sum_{i \in [k]} \sum_{a \in A(G')} c(a) \cdot x_{i,a}&&&\tag{{\bf LP-kTSPP}}\label{lp:ktspp}\\
    {\bf s.t.}\quad x_i(\delta^{out}(s_i)) &= x_i(\delta^{in}(t_i)) = 1&&\qquad\forall~ i \in [k]\nonumber\\
              x_i(\delta^{in}(s_i)) &= x_i(\delta^{out}(t_i)) = 0 &&\qquad\forall~ i \in [k]\nonumber\\
              x_i(\delta^{in}(v)) &= x_i(\delta^{out}(v)) = z_{i,v} &&\qquad\forall~i \in [k], v \in V\setminus T\nonumber\\
              x_i(\delta^{in}(U)) &\ge z_{i,v}&&\qquad\forall~ i \in [k], v \in V\setminus T,  \{v\} \subseteq U \subseteq V\setminus \{s_i\} \nonumber\\
             \sum_{i \in [k]} z_{i,v} & = 1 && \qquad\forall~ v\in V\setminus T\nonumber\\
              x, z &\ge 0 \nonumber
\end{alignat}

To round this LP, we use a result by Bang-Jensen et al.\ about decomposing preflows into branchings. Namely, for a directed graph $D = (V, A)$ with a designated source node $r$, an $r$-preflow is an assignment $x : A \rightarrow \mathbb{Q}_{\geq 0}$ such that $x(\delta^{in}(v)) \geq x(\delta^{out}(v))$ for all $v \in V \setminus \{r\}$.  An $r$-branching in $D$ is directed tree oriented away from the root node $r$ that does not necessarily span all nodes.

\begin{thm}[Bang-Jensen, Frank, and Jackson \cite{BJ95}; Swamy and Post \cite{Post15}]\label{thm:bj}
Let $D = (V, A)$ be a directed graph and $x \in \mathbb R_{\geq 0}^A$ be an $r$-preflow for some $r \in V$. For $v \in V\setminus \{r\}$, let $z_v := \min_{\{v\} \subseteq S \subseteq V\setminus \{r\}} x(\delta^{in}(S))$ be the $r$-$v$ connectivity in $D$ under capacities $x$. We can find $r$-branchings $B_1, \ldots, B_q$
and associated weights $\lambda_1, \ldots, \lambda_q \geq 0$ such that $\sum_i \lambda_i = 1$, $\sum_{i : a \in A(B_i)} \lambda_i \leq x_a$ for each $a \in A$, and $\sum_{i : v \in V(B_i)} \lambda_i \geq z_v$ for each $v \in V$. Moreover, such a decomposition can be computed in time that is polynomial in $|V|$ and the number of bits used to represent $x$.
\end{thm}
Bang-Jensen et al. proved the structural result for integer preflows \cite{BJ95}, Swamy and Post showed how to do this efficiently for rational preflows \cite{Post15}.

Consider an optimal extreme point solution $(x,z)$ to \eqref{lp:ktspp}, in particular every entry in $x$ and $z$ is a rational value with bit complexity being polynomial in the input size.
For any $i$, we have that $x_i$ is an $s_i$-preflow and every $v \in V\setminus T$ is at least $z_{i,v}$-connected from $s_i$ in the directed graph with edge capacities given by $x_i$.
So we can decompose $x_i$ into a convex combination of trees such that each $v \in V\setminus T$ lies on a $z_{i,v}$-fraction of these trees.

Our final algorithm for \kTSPP will flip a coin with some bias for each $i$: in one case it will then sample a tree from this decomposition of $x_i$ and turn it into an $s_i$-$t_i$ path by doubling some edges, and in another case it will just pick the direct $s_i$-$t_i$ edge. After shortcutting paths so each node in $V\setminus T$ lies on at most one path, our final task is to graft in the remaining nodes that were not covered by any tree. If we denote the set of nodes covered collectively by all paths $S$, this algorithm so far can be regarded as having sampled a subset of nodes $S$ which are now ``covered'' and the remaining nodes can be grafted in anywhere. Theorem \ref{thm:bridge} is then used to show the cost of grafting in the remaining nodes is, in expectation, small compared to $OPT$.

To optimize the constants, our algorithm is sensitive to the difference between the optimum LP solution cost and $\sum_{i \in [k]} c(s_it_i)$. This is similar to the notion of {\em regret} or {\em excess} that has been studied in other vehicle routing problems (i.e. how much the solution costs in excess of the the direct path between endpoints). The choice of bias for the coins mentioned above depends on this gap.

The relaxation \eqref{lp:ktspp} can be viewed as a relaxation for \OTSP if we use pairs $(s_i,t_i) := (o_i,o_{i+1})$ for $1 \leq i < k$ and $(s_k,t_k) := (o_k,o_1)$.
Our algorithm uses essentially the same idea, except we always sample a tree for each pair and we do not turn the tree into a path just yet. We are left with a collection of trees whose union contains a cycle passing through $o_1, o_2, \ldots, o_k$ in the required order. Again, Theorem \ref{thm:bridge} allows us to cheaply graft in the remaining nodes (with expected cost being a constant-factor smaller than the optimum cost). This time, we do parity correction by 
optimally pairing the odd-degree nodes as in the Christofides/Serdyukov algorithm for \TSP. A more careful shortcutting argument to turn the resulting Eulerian circuit into a feasible \OTSP argument completes the algorithm. In fact, we present our \OTSP algorithm first since it serves as a good warm up for our \kTSPP algorithm.

\subsection{Notation and Preliminaries}\label{sec:preliminaries}
Throughout the paper, we use the following notation. We define $[k] := \{1,2,\dotsc,k\}$.
For a graph $G$, $V(G)$ is its vertex set and $E(G)$ its edge set.
For a given set of vertices $S \subseteq V$, $G[S]$ denotes the subgraph of $G$ induced by $S$.
We write $G/S$ for the contraction of $S$ in $G$, keeping all loops and parallel edges that are produced as a result (in particular, $G/S$ may not be a metric graph even if $G$ is).

Let  $G = (V,E)$ be an undirected graph with edge costs $c(e) \geq 0, e \in E$, for a nonempty subset $T \subseteq V$ we say a $T$-rooted spanning forest is a forest $F \subseteq E$ such that each component contains exactly one node in $T$.
\begin{lem}\label{lem:mcsf}
A $T$-rooted spanning forest $F$ is a minimum-cost $T$-rooted spanning forest if and only if for each edge $e = uv \in E\setminus F$, the following holds. If $u,v$ are in the same component of $F$, then $c(e) \geq c(e')$ for every edge $e'$ on the unique path between $u$ and $v$ in $F$.
Otherwise, if $u,v$ are not in the same component of $F$, then $c(e) \geq c(e')$ for every edge on the unique path between $t_u$ and $t_v$ in $F \cup \{e\}$ where $t_u, t_v \in T$ are the terminals the components including $u$ and $v$ are rooted at (respectively).
\end{lem}
\begin{proof}
First note that a set of edges $F$ is a minimum-cost $T$-rooted spanning forest if and only $F$ is a minimum-cost spanning tree in $G/T$. Furthermore, a spanning tree in a graph is a minimum-cost spanning tree if and only if every edge $uv$ not on the tree is at least as costly as every edge on the path of the spanning tree between $u$ and $v$. The two cases in the lemma statement for an edge $e \in E\setminus F$ arise from looking at the path between endpoints of $e$ in the minimum spanning tree obtained from $F$ by contracting $T$.
\end{proof}
For directed graphs, we define $T$-rooted directing spanning forests identically: all components contain a single node of $T$ and are directed trees oriented away from this node.

For a set $W \subset V$ of vertices in a directed graph $G = (V,A)$, $\delta^{out}(W) := \{(u,v) \in A \mid u \in W, v \in V\setminus W\}$
is the set of arcs leaving $W$, $\delta^{in}(W) := \{(v,u) \in A \mid u \in W, v \in V\setminus W\}$ is
the set of edges entering $W$, and $\delta(W) := \delta^{out}(W) \cup \delta^{in}(W)$.
We write $\delta^{out}(v)$, $\delta^{in}(v)$, and $\delta(v)$
as shorthand for $\delta^{out}(\{v\})$, $\delta^{in}(\{v\})$, and
$\delta(\{v\})$. 
We often deal with a vector/function $f$ over a finite set $X$. The notation $f(X)$ means $\sum_{x \in X} f(x)$ (or the sum over $f_x$ if $f$ is a vector indexed by $x$).

In the \emph{bidirected cut LP} (BDC) for minimum-cost
directed spanning trees rooted at $r \in V$, we minimize $\sum_{a \in A} c(a)x_{a}$ subject to the following constraints:
\begin{alignat}{2}
     x(\delta^{in}(U)) &\ge 1&&\qquad\forall~ \emptyset \subsetneq U \subseteq V\setminus \{r\}, \nonumber\\
     x &\ge 0. &&\nonumber
\end{alignat}
It is well known that (BDC) is integral~\cite{Edm67}, so the minimum-cost directed spanning tree has cost at most the cost of any feasible solution to (BDC).

We require a slight generalization of this for $T$-rooted directed spanning forests.
\begin{lem}\label{lem:mcdsf}
The polyhedron $\mathcal P_T := \{x \in \mathbb R^{A}_{\geq 0} : x(\delta^{in}(U)) \geq 1$ for all $\emptyset \subsetneq U \subseteq V\setminus T\}$ is integral. As a consequence, for any $x \in \mathcal P_T$ and any costs $c(a)$ over the arcs $a \in A$
we have that $\sum_{a \in A} c(a) \cdot x_a$ upper bounds the cost of a minimum-cost $T$-rooted directed spanning forest.
\end{lem}
\begin{proof}
That the polytope is integral follows directly by considering $x$ in the contracted graph $G/T$ and using integrality of the single-root case and a natural corresponded between rooted directed spanning trees in $G/T$ and $T$-rooted directed spanning forests in $G$.
The second statement is because any integral point in $\mathcal P_T$ contains a $T$-rooted directed spanning forest as a subgraph.
\end{proof}


\section{A Bridge Lemma}\label{sec:bridge}
In this section, we let $G = (V,E)$ be an arbitrary connected and undirected graph with edge costs $c(e) \geq 0, e \in E$. Fix a nonempty subset $T \subseteq V$ and recall $c_T$ denotes the cost of a cheapest $T$-rooted spanning forest.

To prove Lemma \ref{thm:bridge}, we introduce further notation. For a set $S \subseteq V$ define $\drop_T(S) := c_T - c_{T\cup S}$ to be the amount by which the cost of the optimum $T$-rooted spanning forest decreases if we add $S$ to $T$. Note $\drop_T(S) \geq 0$ since we can always get a $(T \cup S)$-rooted spanning forest from a $T$-rooted spanning forest by iteratively deleting an edge on a path between two nodes in $T \cup S$. A {\em fractional covering} of $V\setminus T$ is a pair $(\mathcal S, z)$ where $z \in \mathbb R_{\geq 0}^\mathcal{S}$ satisfies $\sum_{S \ni v} z_S \geq 1$ for each $v \in V\setminus T$.

The following lemma is proven in nearly the same way as the Bridge Lemma in \cite{BGRS13}. Ours can be viewed as restatement that is convenient in our setting. For completeness, we include a proof.
\begin{lem}\label{lem:support}
Let $T \subseteq V$ be nonempty and $(\mathcal S, z)$ a fractional covering of $V\setminus T$. Then $c_T \leq \sum_{S \in \mathcal S} \drop_T(S) \cdot z_S$.
\end{lem}
\begin{proof}
Let $F \subseteq E$ be a minimum-cost $T$-rooted spanning forest.
Build a complete directed graph $G' = (V,A)$ with arc costs $c'(u,v)$ for distinct $u,v \in V$ given as follows. If $u$ and $v$ lie in the same component in the forest $(V,F)$, set $c'(u,v)$ to be the cost of the most expensive edge along the unique $u$-$v$ path in $F$. Otherwise, let $c'(u,v) = \infty$. Note $c'(u,v) = c(u,v)$ for each $uv \in F$.

For each $S \in \mathcal S$, we describe a directed forest $F'_S$ in $G'$. First, let $C \subseteq V$ be a component of the forest $(V,F)$. We know $C \cap (S \cup T)$ contains a unique node $t_C \in T$. Orient the edges $F$ in this component away from $t_C$ and let $F'_S$ be the directed forest rooted at $T$ consisting of arcs $(x,y)$ such that there is a directed path from $x$ to $y$ in this orientation of $F$ of that does not pass through any other node of $C \cap (S \cup T)$. See Figure \ref{fig:orient} for an illustration. 

\begin{figure}[h]
\begin{center}
\includegraphics[width=10cm]{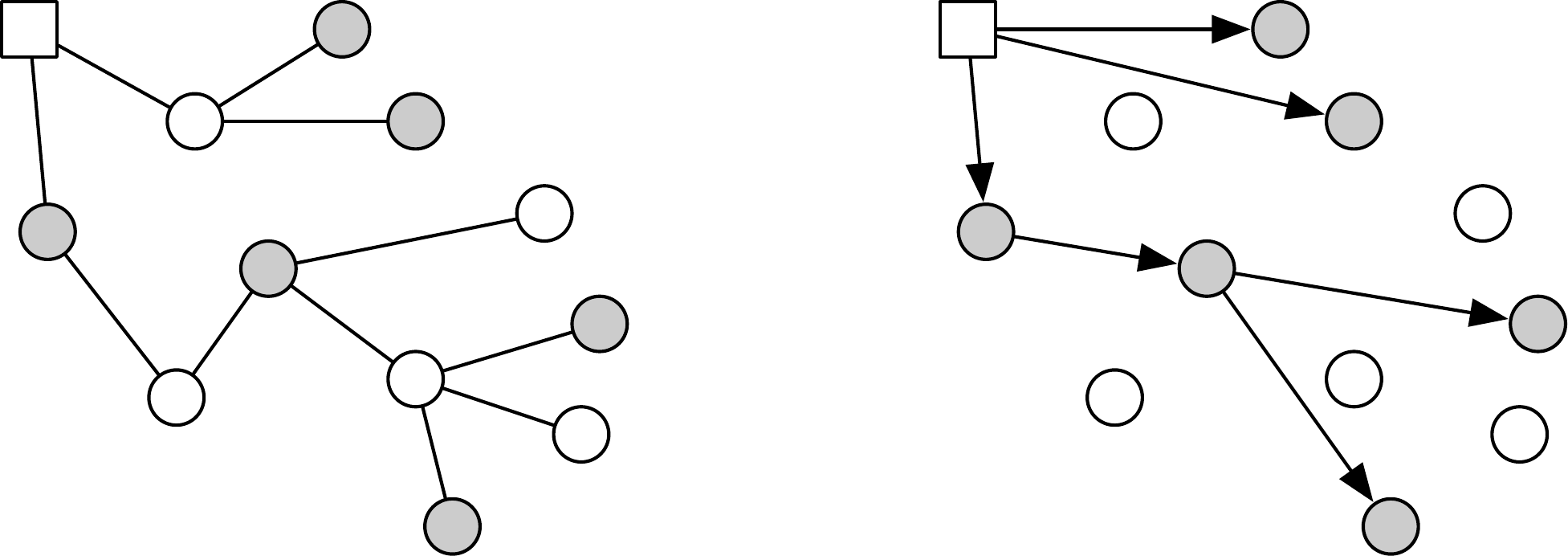}
\end{center}
\caption{Left: A component $C$ of the forest $(V,F)$. The square node is $t_C$ and the nodes in $C \cap S$ are shaded. Right: The corresponding directed tree $F'_S$. The $c'$-cost of each directed edge $(u,v)$ is the largest $c$-cost of edges on the $u$-$v$ path in $F$.}
\label{fig:orient}
\end{figure}

We pause to note that if we delete an edge $e$ with cost $c'(u,v)$ in the path between $u$ and $v$ in $F$ for each arc $(u,v) \in F'_S$, we obtain a minimum-cost $(T \cup S)$-rooted spanning forest in $G$. That is, it is straightforward to verify edges not in the constructed $(T \cup S)$-rooted spanning forest satisfy the properties of Lemma~\ref{lem:mcsf}. That is, $\drop_T(S) = c'(F'_S)$.

Next we define a point $y$ over directed arcs of $A$ by setting
\[ y_a = \sum_{\substack{S \in \mathcal S \\ a \in F'_S}} z_S. \]
That is, $a$ is the extent to which it is generated in the construction of $F'_S$ for various sets $S \in \mathcal S$ in the cover of $V\setminus T$. Observe $\sum_a c'(a) \cdot y_a = \sum_{S \in \mathcal S} z_S \cdot \drop_T(S)$ because $\drop_T(S) = c'(F'_S)$.

We claim for each $\emptyset \subsetneq U \subseteq V\setminus T$ that $y(\delta^{in}(U)) \geq 1$. Consider any $v \in U$, because $(\mathcal S, z)$ is a fractional cover we have $\sum_{S \ni v} z_S \geq 1$. Now, $v$ would be reachable from some $T$ in each set $F'_S$ for $S \in \mathcal S$ with $v \in S$. So $\delta^{in}(U)$ contains an edge along the $T$-$v$ path in $F'_S$ and this edge is counted with weight $z_S$ in $y(\delta^{in}(U))$. Summing over all $S$ containing $v$ finishes this claim.

Let $c'_T$ denote the minimum-cost $T$-rooted spanning forest if we use costs\footnote{While the cost function $c'$ was defined for a directed graph, it is symmetric, so it makes sense to use these costs for an undirected graph.} $c'$. Trivially, $c'_T$ is the same cost as the minimum-cost $T$-rooted spanning forest in $G'$ using costs $c'$.
By Lemma \ref{lem:mcdsf} we have $c'_T$ is at most $\sum_{a \in A} c'(a) \cdot y_a = \sum_{S \in \mathcal S} \drop_T(S) \cdot z_S$.

Our proof concludes by showing $c_T = c'_T$. To that end, let $F^*$ be an optimum $T$-rooted spanning forest under costs $c'$ and among all such rooted spanning forests pick one that maximizes $|F \cap F^*|$. We claim $F = F^*$. This would suffice since $c(e) = c'(e)$ for any $e \in F$ by definition of $c'$.

Otherwise there must be some $uv \in F^*\setminus F$. Since $uv \notin F$ yet $u,v$ lie in the same component of $(V,F)$ (as $F^*$ includes no $\infty$-cost edges) then optimality of $F$ shows $c(u,v) \geq c(e)$ for all edges $e$ on the path $P$ between $u$ and $v$ in $(V,F)$. Fix $e \in P$ to be any edge that bridges the components containing $u$ and $v$ respectively in $(V, F^*\setminus \{uv\})$. Note $c'(e) \leq c'(u,v)$ since $e$ is an edge of $P$ and $c'(uv)$, by definition, is the $c$-cost of the most expensive edge on $P$ and $c(e') = c'(e')$ for any edge $e' \in F$. Thus, $c'(e) \leq c'(u,v)$ means $F^*\setminus \{uv\}\cup\{e\}$ is yet another minimum $c'$-cost $T$-rooted spanning forest contradicting maximality of $|F \cap F^*|$. Thus, $F = F^*$.
\end{proof}

Using this, we complete the proof of our main supporting result which is really just a reinterpretation of Lemma \ref{lem:support} in the setting where 
we randomly sample a subset of $V$ such that nodes in $V-T$ have a uniform lower bound on the probability of being selected.
\begin{proof}[Proof of Theorem \ref{thm:bridge}]
Clearly $c_{T \cup S} \leq c_T$ for any $S \subseteq V$ since we can drop some edges from an optimal $T$-rooted spanning forest to get an $(S \cup T$)-rooted spanning forest. So the result is trivial for $\gamma = 1$. We now consider the case $\gamma < 1$.

Define a fractional set covering $(\mathcal S, z)$ of $V-T$ using the sets $\mathcal S$ supported by the distribution where we set $z_S = \frac{1}{1-\gamma} {\bf Pr}[S]$. By Lemma~\ref{lem:support},
\[ c_T \leq \sum_S  z_S \cdot {\rm drop}_T(S) = \frac{1}{1-\gamma} \cdot \sum_S {\bf Pr}[S] \cdot (c_T - c_{T\cup S}) \]
Rearranging,
\[ \frac{1}{1 - \gamma} {\bf E} [c_{T \cup S}] = \frac{1}{1-\gamma} \cdot \sum_S {\bf Pr}[S] \cdot c_{T \cup S} \leq \left(\frac{1}{1-\gamma} \cdot \sum_S {\bf Pr}[S] \cdot c_T\right) - c_T = \left(\frac{1}{1-\gamma} - 1\right) \cdot c_T = \frac{\gamma}{1-\gamma} \cdot c_T. \]
Multiplying through by $1-\gamma$ finishes the proof.
\end{proof}


\section{Approximating OTSP}\label{sec:approx-OTSP}

Consider an optimal extreme point solution $(x, z)$ to \eqref{lp:ktspp} where we use $(s_i,t_i) = (o_i,o_{i+1})$ for $1 \leq i \leq k$ using wrap-around indexing (i.e. $o_{k+1} := o_1$). We view $t_i$ and $s_{i+1}$ as different copies of $o_i$ for the purpose of solving the LP, but once we sample branchings below we will view them as the same node again.

Using Theorem \ref{thm:bj}, for each $1 \leq i \leq k$ let $B_{i,1}, \ldots, B_{i,q_i}$ be the branchings obtained from decomposing the $r$-preflow with corresponding weights $\lambda_{i,1}, \ldots, \lambda_{i,q_i} \geq 0$ summing to 1 for each $i$.

\begin{algorithm}
\caption{Ordered TSP}
\begin{algorithmic}\label{alg:otsp}
\STATE For each $1 \leq i \leq k$, sample a single branching $B_i$ from the branchings $B_{i,1}, \ldots, B_{i, q_i}$ with probability given by the $\lambda_{i_j}$-values.
\STATE Let $T' = \cup_i V(B_i)$ be the nodes covered by the $B_i$ including all terminals.
\STATE Let $F_{T'}$, the cheapest $T'$-forest rooted at covered nodes.
\STATE Apply Lemma \ref{lem:parity} to convert $\cup_i E(B_i) \cup F_{T'}$ to a feasible OTSP solution.
\end{algorithmic}
\end{algorithm}

We first analyze the expected costs of the branchings and $F_{T'}$.
\begin{lem}
For each $i$, ${\bf E}[c(B_i)] \leq \sum_a c(a) \cdot x_{i,a}$. Overall, the expected total cost of all branchings is at most $OPT_{LP}$.
\end{lem}
\begin{proof}
This is just because the $\lambda_i$-weight of branchings $B_{i,1}, \ldots, B_{i,q_i}$ containing any particular arc $a$ is at most $x_{i,a}$. That is, ${\bf Pr}[a \in B_i] \leq x_{i,a}$ so by linearity of expectation ${\bf E}[c(B_i)] = \sum_a c(a) \cdot {\bf Pr}[a \in B_i] \leq \sum_a c(a) \cdot x_{i,a}$.
Finally, note $OPT_{LP} = \sum_{i \in [k]} \sum_a c(a) \cdot x_{i,a}$.
\end{proof}

\begin{lem}\label{lem:cover}
${\bf E}[c(F_{T'})] \leq e^{-1} \cdot OPT_{LP}$
\end{lem}
\begin{proof}
For each $v \in V\setminus T$ and each $i \in [k]$, ${\bf Pr}[v \in V(B_i)] \geq z_{i,v}$ by properties of the decomposition of $x_i$ mentioned in Theorem \ref{thm:bj}. Since these branchings are sampled independently, a vertex $v$ is not in any branching in the joint distribution over branchings
with probability at most \[ \prod_{i \in [k]} (1-z_{i,v}) \leq (1 - \sum_i z_{i,v}/k)^k = (1-1/k)^k \leq e^{-1} \]
where we have used the arithmetic/geometric mean inequality in the first bound.

By Theorem \ref{thm:bridge}, we have ${\bf E}[c(F_{T'})] = {\bf E}[c_{T'}] \leq e^{-1} \cdot c_T$. To conclude, we show $c_T \leq OPT_{LP}$. Consider the values $x_a := \sum_{i \in [k]} x_{i,a}$ for arcs $a \in A$. Consider any $\emptyset \subsetneq U \subseteq V-T$ and consider any $v \in U$. We have $x(\delta^{in}(U)) = \sum_{i \in [k]} x_i(\delta^{in}(U)) \geq \sum_{i \in [k]} z_{i,v} = 1$ so $x$ is a fractional solution in the bidirected cut relaxation for the minimum-cost $\{s_1, \ldots, s_k\}$-rooted spanning forest LP. So by Lemma \ref{lem:mcdsf}, $c_T \leq \sum_a c(a) \cdot x_a = OPT_{LP}$ as required.
\end{proof}

\begin{lem}\label{lem:parity}
Given branchings $B_1, \ldots, B_k$ where each $B_i$ is an $o_i$-branching that includes node $o_{i+1}$ and given a $T'$-rooted spanning forest $F_{T'}$, we can compute a feasible OTSP solution with cost at most $\sum_i c(B_i) + c(F_{U-V}) + OPT_{LP}/2$ in polynomial time.
\end{lem}
\begin{proof}
    This is essentially the standard parity correction found in the Christifides/Serdyukov algorithm using Wolsey's analysis~\cite{Wol80} for TSP except we have to worry about the resulting tour being a feasible OTSP solution.

For each $B_i$, let $P'_i$ the $o_i$-$o_{i+1}$ tour
in $B_i$. Consider the disjoint union $F'$ of the sets $E(B_1) \setminus E(P'_1), E(B'_2) \setminus E(P'_2), \ldots, E(B'_k) \setminus E(P'_k), F_{T'}$ (i.e. keep all parallel copies of edges). Each component of $F'$ has a node lying on at least one of the paths $P'_i$ since each component in $E(B_i) - E(P'_i)$ has a node on $P'_i$ (as $B_i$ is a branching including $o_i$ and $o_{i+1}$) and each component of $F_{T'}$ is rooted at a node in $T'$. Let $\mathcal O$ be the odd-degree nodes in this multiset $F'$.

Build the following vector $\hat{x} \in \mathbb R^E$ by setting $\hat{x}_{uv} = \sum_{i=1}^k (x_{i, (u,v)} + x_{i, (v,u)})$. We claim $\hat{x}/2$ is a fractional $\mathcal O$-join, this follows by showing $\hat{x}(\delta(S)) \geq 2$ for any cut $\emptyset \subsetneq S \subseteq V$.
If $0 < |S \cap \{o_1, \ldots, o_k\}| < k$, it must be that $o_i \in S, o_{i+1} \notin S$ and $o_j \notin S, o_{j+1} \in S$ for some $i, j$, i.e. consider the point when the wraparound sequence $v_1, \ldots, v_k$ leaves $S$ and when it re-enters $S$.
Since $x_i$ is one unit of $o_i$-$o_{i+1}$ flow then $x_i(\delta^{out}(S)) \geq 1$. Similarly, $x_j(\delta^{in}(S)) \geq 1$ since $x_j$ is one unit of $o_j$-$o_{j+1}$ flow. Thus, $\hat{x}(\delta(S)) \geq x_i(\delta^{out}(S)) + x_i(\delta^{in}(S)) \geq 2$.

Now suppose $S$ does not separate terminals and, by replacing $S$ with $V\setminus S$ if necessary, that $T \subseteq S$. Let $v \in V\setminus S$ be any node. For each $i$, the cut constraints ensure $x_i(\delta^{out}(S)) \geq z_{i,v}$. Since $x_i$ is one unit of $v_i$-$v_{i+1}$ flow and since $v_i, v_{i+1} \in T$ then $x_i(\delta^{in}(S)) = x_i(\delta^{out}(S)) \geq z_{i,v}$ as well. Summing over all $i$ shows
\[ \hat{x}(\delta(U)) = \sum_i x_i(\delta^{out}(S)) + x_i(\delta^{in}(S)) \geq \sum_i z_{i,v} + z_{i,v} = 2. \]

By integrality of the $\mathcal O$-join polyhedron, the minimum-cost $\mathcal O$-join $J$ has $c(J) \leq OPT_{LP}/2$. Adding $J$ to $F'$ yields a multiset of edges that has even degree at each node, so each component can be turned into a cycle by shortcutting without increasing the cost. Each such component includes a node spanned by at least one $P'_i$ since that was true even before adding $J$, so we can graft each cycle into some $P'_i$.
\end{proof}

Putting all these bounds together, we see the expected cost of the solution is at most $(1 + e^{-1} + 1/2)\cdot OPT_{LP} = (3/2 + e^{-1})\cdot OPT_{LP}$ as required.


\section{Approximating $k$-Person TSP Path}

We present this in two steps, the first being a warmup algorithm that achieves an improved approximation for \kTSPP but is not as good as our final algorithm.

Let $\Delta := \sum_{i \in [k]} c(s_it_i)$. Since each $x_i$ constitutes one unit of $s_i$-$t_i$ flow, we have $\sum_a c(a) \cdot x_{i,a} \geq c(s_it_i)$. Summing over all $i$ shows $\Delta \leq OPT_{LP}$. Finally, let $\tau \in [0,1]$ be such that $(1-\tau) \cdot OPT_{LP} = \sum_{i \in [k]} c(s_it_i)$.

\subsection{Warmup Algorithm}

As with \OTSP, we consider an extreme point optimal solution $(x,z)$ to \eqref{lp:ktspp}. For each $i \in [k]$, let $B_{i,1}, \ldots, B_{i, q_i}$ be branchings and $\lambda_{i,1}, \ldots, \lambda_{i,q_i}$ as in Theorem \ref{thm:bj}.

\begin{algorithm}
\caption{$k$-Person TSP Path -- Warmup Algorithm}
\begin{algorithmic}\label{alg:ktspp1}
\STATE For each $1 \leq i \leq k$, sample a single branching $B_i$ from $B_{i,1}, \ldots, B_{i, q_i}$ with corresponding probabilities given by the $\lambda_{i,j}$-values.
\STATE For each $1 \leq i \leq k$, let $P_i$ be an $s_i$-$t_i$ path obtained by doubling all edges not on the $s_i$-$t_i$ path in $B_i$ and shortcutting.
\STATE Let $T' = \cup_i V(P_i)$ be the covered nodes including the original terminals $T$.
\STATE Let $F_{T'}$, the cheapest $T'$ forest rooted at covered nodes.
\STATE Double the edges of $F_{T'}$ and shortcut to get simple cycles, graft these cycles into the paths $P_1, \ldots, P_k$ to get the final solution.
\end{algorithmic}
\end{algorithm}
In the final step, each component of $F_{T'}$ includes a single node $t' \in T'$. The resulting cycle can then be grafted in to any path $P_i$ that covers $t'$.

The final cost is at most $\sum_{i \in [k]} c(P_i) + 2 \cdot c(F_{T'})$. We analyze the expected costs of these portions of the solution. For brevity, let $c(x_i)$ denote $\sum_a c(a) \cdot x_{i,a}$ denote the contribution of the $s_i$-$t_i$ flow $x_i$ to $OPT_{LP}$.
\begin{lem}\label{lem:pathcost1}
For each $i \in [k]$, ${\bf E}[c(P_i)] \leq 2 \cdot c(x_i) - c(s_i,t_i)$.
\end{lem}
\begin{proof}
Fix such an $i$. We have $c(P_i) \leq 2 \cdot c(B_i) - c(s_i,t_i)$ because $P_{i}$ is obtained by doubling all edges of $B_i$ except those on an $s_i$-$t_i$ path and then shortcutting (which does not increase the cost, by the triangle inequality).
Thus, 
\begin{align*}
{\bf E}[P_i] &\leq \sum_j \lambda_{i,j} \cdot (2 \cdot c(B_{i,j}) - c(s_i,t_i)) \\
&=  2 \cdot \sum_j \lambda_{i,j} \cdot c(B_{i,j}) - \sum_j \lambda_{i,j} \cdot c(s_i,t_i) \\
&\leq  2 \cdot c(x_i) - c(s_i,t_i)
\end{align*}
The last bound is because $\sum_j \lambda_{i,j} = 1$ and because each arc $a$ appears on at most a total weight of $x_{i,a}$ of branchings, by Theorem \ref{thm:bj}.
\end{proof}

\noindent
The proof of the following is identical to the proof of Lemma \ref{lem:cover} since the path $P_i$ spans the same set of nodes as the branching $B_i$ and the branchings were sampled the same way as in Algorithm \eqref{alg:otsp}.
\begin{lem}
${\bf E}[c(F_{T'})] \leq e^{-1} \cdot OPT_{LP}$
\end{lem}

Putting these bounds together shows:
\begin{align*}
{\bf E}\left[\sum_{i \in [k]} c(P_i) + 2 \cdot c(F_{T'})\right] &=  \sum_{i \in [k]} {\bf E}[P_i] + 2 \cdot c(F_{T'}) \\
&\leq  \sum_i (2 \cdot c(x_i)-c(s_i,t_i)) + 2  e^{-1} \cdot OPT_{LP} \\
&=  2 \cdot OPT_{LP} - \Delta + 2 e^{-1} \cdot OPT_{LP} \\
&=  2 \cdot OPT_{LP} - (1-\tau) \cdot OPT_{LP} + 2 e^{-1} \cdot OPT_{LP} \\
&=  (1+\tau + 2 e^{-1}) \cdot OPT_{LP}
\end{align*}
Note for any $\tau \in [0,1]$ the approximation guarantee is at most $2 + 2 e^{-1} < 2.7358$ and this worst case is attained only for large values of $\tau$.

On the other hand, for large $\tau$ a simpler algorithm performs much better.
\begin{lem}\label{lem:simple}
Consider the algorithm that doubles a minimum-cost $T$-rooted spanning forest and then adds the direct edges $s_it_i$ to get $s_i$-$t_i$ paths spanning all nodes (after shortcutting). The cost of this solution is at most $(3-\tau) \cdot OPT_{LP}$.
\end{lem}
\begin{proof}
The cost of the cycles is at most $2 \cdot c_T \leq 2 \cdot OPT_{LP}$ and direct edge costs are $\Delta = (1-\tau) \cdot OPT_{LP}$.
\end{proof}

Finally, if we output the better of these two algorithms the approximation guarantee is no worse than $\max_{\tau \in [0,1]} \min\{1 + \tau + 2 \cdot e^{-1}, 3-\tau\}$. The worst case occurs when $\tau = 1-e^{-1}$ which yields a notably better approximation guarantee
of $2+e^{-1} < 2.368$.


\subsection{An Improved \kTSPP Approximation}

Intuitively, the Algorithm \ref{alg:ktspp1} is better for small $\tau$ because turning the branchings into paths results in only a small increase in the cost (in expectation) as most of the cost of the branchings is usually concentrated on their $s_i$-$t_i$ paths.
The Algorithm described in Lemma \ref{lem:simple} does better for large $\tau$ because the direct $s_i$-$t_i$ connections are cheap. We can do better with a single algorithm that interpolates between these two ideas before grafting in the remaining nodes.

That is, we will sample an $s_i$-$t_i$ path $P_i$ by flipping a coin: in one case we will sample a branching $B_i$ and turn it into a path like the $(2+2e^{-1})$-approximation in Algorithm \eqref{alg:ktspp1}, in the other case we will just add the direct edge $s_it_i$ as in $(3-\tau)$-approximation mentioned in Lemma \ref{lem:simple}. In either case, we then double a minimum-cost $T'$-rooted spanning forest to graft in the remaining nodes. The bias of the coin will depend on $\tau$, which measures the gap between $\Delta$ and $OPT_{LP}$. The full algorithm is described in Algorithm \ref{alg:ktspp2}.

\begin{algorithm}
\caption{$k$-Person TSP Path -- Final Algorithm}
\begin{algorithmic}\label{alg:ktspp2}
\STATE Let $\gamma = \min\{1, \ln \tau^{-1}\}$. \COMMENT{Using $\gamma = 1$ if $\tau = 0$}
\STATE For each $1 \leq i \leq k$, sample a single branching $B_i$ from $B_{i,1}, \ldots, B_{i, q_i}$ with corresponding probabilities given by the $\lambda_i$-values.
\STATE Independently for each $1 \leq i \leq k$, with probability $\gamma$ let $P_i$ be an $s_i$-$t_i$ path obtained by doubling all edges not on the $s_i$-$t_i$ path in $B_i$ and shortcutting, otherwise let $P_i$ just be the single edge $s_it_i$.
\STATE Let $T' = \cup_i V(P_i)$ be the covered nodes including the original terminals $T$.
\STATE Let $F_{T'}$, the cheapest $T'$ forest rooted at covered nodes.
\STATE Double the edges of $F_{T'}$ and shortcut to get simple cycles, graft these cycles into the paths $P_1, \ldots, P_k$ to get the final solution.
\end{algorithmic}
\end{algorithm}

\begin{lem}\label{lem:pathcost2}
${\bf E}\left[\sum_{i \in [k]} c(P_i)\right] \leq (1 - \tau + 2\gamma\tau) \cdot OPT_{LP}.$
\end{lem}
\begin{proof}
Conditioned on the event that $P_i$ is obtained by turning $B_i$ into a path, the expected cost of $P_i$ would be at most $2 \cdot c(x_i) - c(s_it_i)$ as in the proof of Lemma \ref{lem:pathcost1}. Removing this conditioning (i.e. considering the entire distribution over $P_i$),
the expected cost is then at most $\gamma \cdot (2 \cdot c(x_i) - c(s_i,t_i)) + (1-\gamma) \cdot c(s_i,t_i)$.

Summing over all $i$ shows
\begin{align*}
{\bf E}\left[\sum_{i \in [k]} c(P_i)\right] &\leq  \sum_{i \in [k]} \left(\gamma \cdot (2 \cdot c(x_i) - c(s_it_i)) + (1-\gamma) \cdot c(s_i,t_i)\right) \\
&=  \sum_{i \in [k]} (2\gamma \cdot c(x_i) + (1-2\gamma) \cdot c(s_i,t_i)) \\
&=  2\gamma \cdot OPT_{LP} + (1-2\gamma) \cdot \Delta \\
&=  2\gamma \cdot OPT_{LP} + (1-2\gamma) \cdot (1-\tau) \cdot OPT_{LP} \\
&=  (1 - \tau + 2\gamma\tau) \cdot OPT_{LP}
\end{align*}
\end{proof}

\begin{lem}
${\bf E}[c(F_{T'})] \leq e^{-\gamma}.$
\end{lem}
\begin{proof}
As in the proof of Lemma \ref{lem:cover}, it suffices to show a node $v \in V\setminus T$ does not lie on any path with probability at most $e^{-\gamma}$.
The probability $v$ lies on $P_i$ is at least $\gamma \cdot z_{i,v}$ since we use the branching $B_i$ to get $P_i$ with probability $\gamma$, and $B_i$ contains $v$ with probability at least $z_{i,v}$. Thus,
\[ {\bf Pr}[v \notin T'] \leq \prod_{i \in [k]}(1-\gamma \cdot z_{i,v}) \leq (1-\gamma \sum_{i \in [k]} \cdot z_{i,v}/k)^k = (1-\gamma/k)^k \leq e^{-\gamma}. \]
\end{proof}

The cost of the final solution is $2 \cdot c(F_{T'}) + \sum_{i \in [k]} c(P_i)$ which, in expectation, is at most $(1 - \tau + 2\gamma\tau + 2e^{-\gamma}) \cdot OPT_{LP}$. If $\tau < e^{-1}$, then $\gamma = 1$
and this guarantee is at most $1 + 3 \cdot e^{-1} < 2.104$. Otherwise, if $\tau > e^{-1}$ then $\gamma = \ln \tau^{-1}$ and the guarantee becomes $1 + \tau - 2 \tau \ln \tau$.
The worst case occurs at $\tau = e^{-1/2}$ with an approximation guarantee of $1 + 2 \cdot e^{-1/2} < 2.2131$. So for all cases of $\tau \in [0,1]$ we have the approximation guarantee being at most $2.2131$.


\section{Conclusion}

We have demonstrated how an adaptation of the Bridge Lemma along with the decomposition of preflows by Bang-Jensen et al. can be used to obtain improved vehicle routing approximations.
Intuitively, this approach seems most effective when the goal is to minimize the total length of a collection of paths that collectively cover all nodes. So one naturally wonders if we can get improvements for \kTSPP when $s_i = t_i$ for all 
pairs (i.e. find rooted cycles covering all nodes). Beating the trivial 2-approximation when $k$ is part of the input (i.e. not fixed) is a compelling open problem. We have been unable to adapt our approach to this setting, but perhaps a variation
of our ideas would succeed.


\section*{Acknowledgments}
We thank Ruben Hoeksma, Nidia Obscura Acosta, José A. Soto, Kevin Schewior and José Verschae for helpful discussions. Progress on this project was made at the 27th Aussois Combinatorial Optimization Workshop, we thank the organizers.

\bibliographystyle{plain}
\bibliography{refs}

\begin{thebibliography}{10}

\bibitem{AMN24}
Susanne Armbruster, Matthias Mnich, and Martin N{\"{a}}gele.
\newblock A (3/2+1/e)-approximation algorithm for ordered {TSP}.
\newblock {\em CoRR}, abs/2405.06244, 2024.

\bibitem{BJ95}
J{\o}rgen Bang-Jensen, Andr{\'a}s Frank, and Bill Jackson.
\newblock Preserving and increasing local edge-connectivity in mixed graphs.
\newblock {\em SIAM J. Discret. Math.}, 8:155--178, 1995.

\bibitem{BCGL07}
Gerardo Berbeglia, Jean-Fran{\c{c}}ois Cordeau, Irina Gribkovskaia, and Gilbert
  Laporte.
\newblock Static pickup and delivery problems: a classification scheme and
  survey.
\newblock {\em Top}, 15(1):1--31, 2007.

\bibitem{BKN00}
Randeep Bhatia, Samir Khuller, and Joseph Naor.
\newblock The loading time scheduling problem.
\newblock {\em J. Algorithms}, 36(1):1--33, 2000.

\bibitem{BMRS94}
Lucio Bianco, Aristide Mingozzi, Salvatore Ricciardelli, and Massimo Spadoni.
\newblock Exact and heuristic procedures for the traveling salesman problem
  with precedence constraints, based on dynamic programming.
\newblock {\em INFOR: Information Systems and Operational Research},
  32(1):19--32, 1994.

\bibitem{BHKK06}
Hans{-}Joachim B{\"{o}}ckenhauer, Juraj Hromkovic, Joachim Kneis, and Joachim
  Kupke.
\newblock On the approximation hardness of some generalizations of {TSP}.
\newblock In {\em {SWAT}}, volume 4059 of {\em Lecture Notes in Computer
  Science}, pages 184--195. Springer, 2006.

\bibitem{BMS13}
Hans{-}Joachim B{\"{o}}ckenhauer, Tobias M{\"{o}}mke, and Monika
  Steinov{\'{a}}.
\newblock Improved approximations for {TSP} with simple precedence constraints.
\newblock {\em J. Discrete Algorithms}, 21:32--40, 2013.

\bibitem{BS14}
Hans{-}Joachim B{\"{o}}ckenhauer and Monika Steinov{\'{a}}.
\newblock Improved approximations for ordered {TSP} on near-metric graphs.
\newblock {\em {RAIRO} - Theor. Inf. and Applic.}, 48(5):479--494, 2014.

\bibitem{BGRS13}
Jaroslaw Byrka, Fabrizio Grandoni, Thomas Rothvo{\ss}, and Laura Sanit{\`{a}}.
\newblock Steiner tree approximation via iterative randomized rounding.
\newblock {\em J. {ACM}}, 60(1):6:1--6:33, 2013.

\bibitem{CMRS97}
Moses Charikar, Rajeev Motwani, Prabhakar Raghavan, and Craig Silverstein.
\newblock Constrained {TSP} and low-power computing.
\newblock In {\em Workshop on Algorithms and Data Structures}, pages 104--115.
  Springer, 1997.

\bibitem{CR98}
Moses Charikar and Balaji Raghavachari.
\newblock The finite capacity dial-a-ride problem.
\newblock In {\em Proceedings 39th Annual Symposium on Foundations of Computer
  Science (Cat. No. 98CB36280)}, pages 458--467. IEEE, 1998.

\bibitem{Chr76}
Nicos Christofides.
\newblock Worst-case analysis of a new heuristic for the travelling salesman
  problem.
\newblock Technical Report 388, Graduate School of Industrial Administration,
  Carnegie-Mellon University, 1976.

\bibitem{CL07}
Jean-Fran{\c{c}}ois Cordeau and Gilbert Laporte.
\newblock The dial-a-ride problem: models and algorithms.
\newblock {\em Annals of operations research}, 153(1):29--46, 2007.

\bibitem{Depp23}
Max Deppert, Matthias Kaul, and Matthias Mnich.
\newblock A $(3/2 + \epsilon)$-approximation for multiple tsp with a variable
  number of depots.
\newblock In Inge~Li G{\o}rtz, Martin Farach-Colton, Simon~J. Puglisi, and
  Grzegorz Herman, editors, {\em 31st Annual European Symposium on Algorithms
  (ESA 2023)}, volume 274 of {\em Leibniz International Proceedings in
  Informatics (LIPIcs)}, pages 39:1--39:15, Dagstuhl, Germany, 2023. Schloss
  Dagstuhl -- Leibniz-Zentrum f{\"u}r Informatik.

\bibitem{Edm67}
Jack Edmonds et~al.
\newblock Optimum branchings.
\newblock {\em Journal of Research of the national Bureau of Standards B},
  71(4):233--240, 1967.

\bibitem{Gian17}
A.~Giannakos, M.~Hifi, R.~Kheffache, and R.~Ouafi.
\newblock An approximation algorithm for the k-fixed depots problem.
\newblock {\em Computers \& Industrial Engineering}, 111:50--55, 2017.

\bibitem{GW95}
Michel~X. Goemans and David~P. Williamson.
\newblock A general approximation technique for constrained forest problems.
\newblock {\em {SIAM} J. Comput.}, 24(2):296--317, 1995.

\bibitem{GNR09}
Inge~Li G{\o}rtz, Viswanath Nagarajan, and R~Ravi.
\newblock Minimum makespan multi-vehicle dial-a-ride.
\newblock In {\em European Symposium on Algorithms}, pages 540--552. Springer,
  2009.

\bibitem{Jai01}
Kamal Jain.
\newblock A factor 2 approximation algorithm for the generalized steiner
  network problem.
\newblock {\em Combinatorica}, 21(1):39--60, 2001.

\bibitem{OG1}
Anna~R. Karlin, Nathan Klein, and Shayan~Oveis Gharan.
\newblock A (slightly) improved approximation algorithm for metric {TSP}.
\newblock In Samir Khuller and Virginia~Vassilevska Williams, editors, {\em
  {STOC} '21: 53rd Annual {ACM} {SIGACT} Symposium on Theory of Computing,
  Virtual Event, Italy, June 21-25, 2021}, pages 32--45. {ACM}, 2021.

\bibitem{OG2}
Anna~R. Karlin, Nathan Klein, and Shayan~Oveis Gharan.
\newblock A deterministic better-than-3/2 approximation algorithm for metric
  {TSP}.
\newblock In Alberto~Del Pia and Volker Kaibel, editors, {\em Integer
  Programming and Combinatorial Optimization - 24th International Conference,
  {IPCO} 2023, Madison, WI, USA, June 21-23, 2023, Proceedings}, volume 13904
  of {\em Lecture Notes in Computer Science}, pages 261--274. Springer, 2023.

\bibitem{LRS11}
Lap~Chi Lau, Ramamoorthi Ravi, and Mohit Singh.
\newblock {\em Iterative methods in combinatorial optimization}, volume~46.
\newblock Cambridge University Press, 2011.

\bibitem{Post15}
Ian Post and Chaitanya Swamy.
\newblock Linear programming-based approximation algorithms for multi-vehicle
  minimum latency problems (extended abstract).
\newblock In Piotr Indyk, editor, {\em Proceedings of the Twenty-Sixth Annual
  {ACM-SIAM} Symposium on Discrete Algorithms, {SODA} 2015, San Diego, CA, USA,
  January 4-6, 2015}, pages 512--531. {SIAM}, 2015.

\bibitem{Sal19}
Yaroslav Salii.
\newblock Revisiting dynamic programming for precedence-constrained traveling
  salesman problem and its time-dependent generalization.
\newblock {\em European Journal of Operational Research}, 272(1):32--42, 2019.

\bibitem{SS95}
Martin~WP Savelsbergh and Marc Sol.
\newblock The general pickup and delivery problem.
\newblock {\em Transportation science}, 29(1):17--29, 1995.

\bibitem{Ser78}
A.~Serdyukov.
\newblock On some extremal walks in graphs.
\newblock {\em Upravlyaemye systemy}, 17:76 -- 79, 1978.

\bibitem{PTSP22}
Vera Traub, Jens Vygen, and Rico Zenklusen.
\newblock Reducing path tsp to tsp.
\newblock {\em SIAM Journal on Computing}, 51(3):STOC20--24--STOC20--53, 2022.

\bibitem{WGMV95}
David~P. Williamson, Michel~X. Goemans, Milena Mihail, and Vijay~V. Vazirani.
\newblock A primal-dual approximation algorithm for generalized steiner network
  problems.
\newblock {\em Combinatorica}, 15(3):435--454, 1995.

\bibitem{Wol80}
Laurence~A. Wolsey.
\newblock Heuristic analysis, linear programming and branch and bound.
\newblock In {\em Combinatorial Optimization II}, volume~13 of {\em
  Mathematical Programming Studies}, pages 121--134. 1980.

\bibitem{Xu10}
Zhou Xu and Brian Rodrigues.
\newblock A 3/2-approximation algorithm for multiple depot multiple traveling
  salesman problem.
\newblock In {\em {SWAT}}, volume 6139, pages 127--138, 06 2010.

\end{thebibliography}

\end{document}